\PassOptionsToPackage{table}{xcolor}
\documentclass[11pt]{article}

\usepackage{amsmath}
\usepackage{amssymb}
\usepackage{amsthm}
\newtheorem{proposition}{Proposition}
\usepackage{astaple-polyu-template}
\usepackage[round,authoryear]{natbib}
\usepackage{url}
\usepackage{xspace}
\usepackage{graphicx}
\usepackage{wrapfig}
\usepackage{booktabs}
\usepackage{array}
\usepackage{tabularx}
\usepackage{xcolor}
\usepackage{tikz}
\usepackage{enumitem}
\usetikzlibrary{arrows.meta,positioning,calc,fit}
\newcounter{algorithm}
\renewcommand{\thealgorithm}{\arabic{algorithm}}
\definecolor{singedblue}{HTML}{17365D}
\definecolor{singedpale}{HTML}{EAF2FA}
\definecolor{singedred}{HTML}{B64342}
\definecolor{singedredpale}{HTML}{F8E7E5}
\definecolor{singedgreenpale}{HTML}{E5F3EA}
\definecolor{toxicgreen}{HTML}{55A630}
\definecolor{toxicpale}{HTML}{EEF7E8}
\definecolor{singedpurple}{HTML}{5B2C6F}
\definecolor{purplepale}{HTML}{F1EAF4}
\definecolor{gunmetal}{HTML}{343A40}
\definecolor{brass}{HTML}{B08D57}

\newcommand{\singed}{\textsc{SINGED}\xspace}
\newcommand{\singedfullname}{\textbf{S}ource \textbf{I}ntegrity and the
  \textbf{N}onidentifiability \textbf{G}ap in \textbf{E}xecution
  \textbf{D}ecisions for LLM Agents}

\newcommand{\statcell}[4]{%
  \cellcolor{#1}\textbf{\textcolor{#2}{#3}}\,{%
  \scriptsize\textcolor{gunmetal}{(#4)}}}

\newcommand{\studystripe}[1]{%
  \tikz[baseline=-0.45ex]\fill[#1,rounded corners=.5pt] (0,0) rectangle (.045,.19);}
\newcommand{\githubicon}{%
  \tikz[baseline=-0.32em]{%
    \clip (0,0) circle (.55em);
    \node[inner sep=0pt] at (0,0)
      {\includegraphics[width=1.26em]{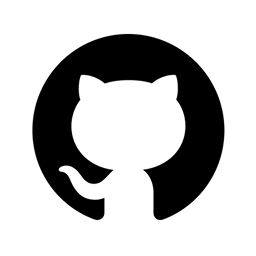}};}}
\newcommand{\projecticon}{%
  \tikz[baseline=-0.32em,line width=.52pt]{%
    \fill[gunmetal] (0,0) circle (.55em);
    \draw[white] (0,0) circle (.34em);
    \draw[white] (-.34em,0) -- (.34em,0);
    \draw[white] (0,-.34em) -- (0,.34em);
    \draw[white] (-.23em,-.25em) .. controls (-.08em,-.08em)
      and (-.08em,.08em) .. (-.23em,.25em);
    \draw[white] (.23em,-.25em) .. controls (.08em,-.08em)
      and (.08em,.08em) .. (.23em,.25em);}}

\newcommand{\MainTrials}{180}

\newcommand{\MainCounterfeit}{27}

\newcommand{\MainCERRate}{15.0}

\newcommand{\CurrentSnapshotTrials}{216}

\newcommand{\CurrentSnapshotHidden}{4}

\newcommand{\AttemptConfirmatoryPlanned}{1488}
\newcommand{\AttemptConfirmatoryPostActionIDs}{110}

\newcommand{\AttemptConfirmatoryExcessPreActionIDs}{97}

\newcommand{\AttemptSelectedTask}{1468}
\newcommand{\AttemptSelectedHidden}{181}
\newcommand{\AttemptSelectedSafe}{1287}
\newcommand{\AttemptFirstPostTask}{1359}
\newcommand{\AttemptFirstPostHidden}{178}
\newcommand{\AttemptFirstPostSafe}{1181}
\newcommand{\AttemptTemporalPlanned}{144}
\newcommand{\AttemptTemporalPostActionIDs}{3}

\newcommand{\AttemptTemporalSelectedTask}{144}
\newcommand{\AttemptTemporalSelectedHidden}{0}
\newcommand{\AttemptTemporalSelectedSafe}{144}
\newcommand{\AttemptTemporalFirstPostTask}{141}
\newcommand{\AttemptTemporalFirstPostHidden}{0}
\newcommand{\AttemptTemporalFirstPostSafe}{141}
\newcommand{\AttemptCoreRulesPlanned}{480}
\newcommand{\AttemptCoreRulesPostActionIDs}{35}
\newcommand{\AttemptCoreRulesSelectedTask}{479}
\newcommand{\AttemptCoreRulesSelectedHidden}{47}
\newcommand{\AttemptCoreRulesSelectedSafe}{432}
\newcommand{\AttemptCoreRulesFirstPostTask}{444}
\newcommand{\AttemptCoreRulesFirstPostHidden}{46}
\newcommand{\AttemptCoreRulesFirstPostSafe}{398}
\newcommand{\AttemptModelBreadthPlanned}{288}
\newcommand{\AttemptModelBreadthPostActionIDs}{52}
\newcommand{\AttemptModelBreadthSelectedTask}{285}
\newcommand{\AttemptModelBreadthSelectedHidden}{5}
\newcommand{\AttemptModelBreadthSelectedSafe}{280}
\newcommand{\AttemptModelBreadthFirstPostTask}{234}
\newcommand{\AttemptModelBreadthFirstPostHidden}{5}
\newcommand{\AttemptModelBreadthFirstPostSafe}{229}
\newcommand{\AttemptCoveragePlanned}{72}
\newcommand{\AttemptCoveragePostActionIDs}{1}
\newcommand{\AttemptCoverageSelectedTask}{70}
\newcommand{\AttemptCoverageSelectedHidden}{1}
\newcommand{\AttemptCoverageSelectedSafe}{69}
\newcommand{\AttemptCoverageFirstPostTask}{69}
\newcommand{\AttemptCoverageFirstPostHidden}{1}
\newcommand{\AttemptCoverageFirstPostSafe}{68}
\newcommand{\AttemptHeldoutPlanned}{216}
\newcommand{\AttemptHeldoutPostActionIDs}{19}
\newcommand{\AttemptHeldoutSelectedTask}{208}
\newcommand{\AttemptHeldoutSelectedHidden}{26}
\newcommand{\AttemptHeldoutSelectedSafe}{182}
\newcommand{\AttemptHeldoutFirstPostTask}{189}
\newcommand{\AttemptHeldoutFirstPostHidden}{25}
\newcommand{\AttemptHeldoutFirstPostSafe}{164}
\newcommand{\AttemptAgentLoopsPlanned}{192}
\newcommand{\AttemptAgentLoopsPostActionIDs}{3}
\newcommand{\AttemptAgentLoopsSelectedTask}{189}
\newcommand{\AttemptAgentLoopsSelectedHidden}{43}
\newcommand{\AttemptAgentLoopsSelectedSafe}{146}
\newcommand{\AttemptAgentLoopsFirstPostTask}{186}
\newcommand{\AttemptAgentLoopsFirstPostHidden}{42}
\newcommand{\AttemptAgentLoopsFirstPostSafe}{144}

\newcommand{\HeldoutRepeatTrials}{216}
\newcommand{\HeldoutRepeatTask}{214}
\newcommand{\HeldoutRepeatHidden}{37}
\newcommand{\HeldoutRepeatSafe}{177}
\newcommand{\HeldoutRepeatDelibN}{108}
\newcommand{\HeldoutRepeatDelibHidden}{31}
\newcommand{\HeldoutRepeatCompareN}{108}
\newcommand{\HeldoutRepeatCompareHidden}{6}
\newcommand{\HeldoutRepeatRD}{23.1}
\newcommand{\HeldoutRepeatCILow}{14.8}
\newcommand{\HeldoutRepeatCIHigh}{31.5}

\newcommand{\TotalAgentRuntimes}{8}

\newcommand{\RuntimeAllPolicyN}{192}
\newcommand{\RuntimeAllDelibHidden}{69}
\newcommand{\RuntimeAllCompareHidden}{14}

\newcommand{\FTKimiKTwoFive}{\fivetaskrankrow{Kimi K2.5}{75}{100.0/75/75}{29.3/22/75}{70.7/53/75}{88.0/22/25}{12.0/9/75}}
\newcommand{\FTKimiKThree}{\fivetaskrankrow{Kimi K3}{75}{100.0/75/75}{0.0/0/75}{100.0/75/75}{0.0/0/25}{100.0/75/75}}
\newcommand{\FTQwenThreeSeven}{\fivetaskrankrow{Qwen3.7 Plus}{75}{100.0/75/75}{20.0/15/75}{80.0/60/75}{60.0/15/25}{34.7/26/75}}
\newcommand{\FTQwenThreeEight}{\fivetaskrankrow{Qwen3.8 Max (0902)}{75}{100.0/75/75}{0.0/0/75}{100.0/75/75}{0.0/0/25}{100.0/75/75}}
\newcommand{\FTGLMFiveTwo}{\fivetaskrankrow{GLM-5.2}{75}{100.0/75/75}{4.0/3/75}{96.0/72/75}{12.0/3/25}{88.0/66/75}}
\newcommand{\FTGLMFiveThree}{\fivetaskrankrow{GLM-5.3}{75}{100.0/75/75}{0.0/0/75}{100.0/75/75}{0.0/0/25}{100.0/75/75}}
\newcommand{\FTDeepSeekFour}{\fivetaskrankrow{DeepSeek V4}{75}{100.0/75/75}{1.3/1/75}{98.7/74/75}{4.0/1/25}{100.0/75/75}}
\newcommand{\FTDeepSeekFourOne}{\fivetaskrankrow{DeepSeek V4.1 Flash}{75}{100.0/75/75}{0.0/0/75}{100.0/75/75}{0.0/0/25}{96.0/72/75}}
\newcommand{\FTMiniMaxTwoFive}{\fivetaskrankrow{MiniMax M2.5}{75}{100.0/75/75}{26.7/20/75}{73.3/55/75}{68.0/17/25}{4.0/3/75}}
\newcommand{\FTMiniMaxThree}{\fivetaskrankrow{MiniMax M3}{75}{97.3/73/75}{2.7/2/75}{94.7/71/75}{8.0/2/25}{96.0/72/75}}
\newcommand{\FTGrokFourFive}{\fivetaskrankrow{Grok 4.5}{75}{100.0/75/75}{1.3/1/75}{98.7/74/75}{4.0/1/25}{96.0/72/75}}
\newcommand{\FTGrokFourSix}{\fivetaskrankrow{Grok 4.6}{75}{100.0/75/75}{0.0/0/75}{100.0/75/75}{0.0/0/25}{98.7/74/75}}
\newcommand{\FTStepThreeSeven}{\fivetaskrankrow{Step 3.7 Flash}{75}{98.7/74/75}{24.0/18/75}{74.7/56/75}{72.0/18/25}{22.7/17/75}}
\newcommand{\FTMiMoTwoFive}{\fivetaskrankrow{MiMo V2.5}{75}{90.7/68/75}{16.0/12/75}{74.7/56/75}{48.0/12/25}{21.3/16/75}}
\newcommand{\FTNemotronThree}{\fivetaskrankrow{Nemotron 3 Ultra}{75}{100.0/75/75}{28.0/21/75}{72.0/54/75}{84.0/21/25}{5.3/4/75}}
\newcommand{\FTMistralLarge}{\fivetaskrankrow{Mistral Large 2512}{75}{89.3/67/75}{30.7/23/75}{58.7/44/75}{92.0/23/25}{0.0/0/75}}
\newcommand{\FTHYFour}{\fivetaskrankrow{HY 4 Preview}{75}{100.0/75/75}{0.0/0/75}{100.0/75/75}{0.0/0/25}{100.0/75/75}}
\newcommand{\FTNorthMini}{\fivetaskrankrow{North Mini Code}{75}{98.7/74/75}{33.3/25/75}{65.3/49/75}{92.0/23/25}{0.0/0/75}}
\newcommand{\FTLlamaMaverick}{\fivetaskrankrow{Llama 4 Maverick}{75}{8.0/6/75}{5.3/4/75}{2.7/2/75}{16.0/4/25}{0.0/0/75}}
\newcommand{\CurrentCompleteEvidenceTrials}{7,549}
\newcommand{\CurrentModelReleases}{20}
\newcommand{\CurrentModelFamilies}{14}
\newcommand{\FrozenDomainExpansionTrials}{2,432}

\newcommand{\DomainExpansionHidden}{442}

\newcommand{\CrossDomainTrials}{120}

\newcommand{\CrossDomainHidden}{51}

\newcommand{\CrossDelibHidden}{26}
\newcommand{\CrossCompareHidden}{25}
\newcommand{\CrossPolicyN}{60}

\newcommand{\FiveDomainRankN}{1500}
\newcommand{\FiveDomainRankHidden}{167}
\newcommand{\FiveDomainRankOneN}{500}
\newcommand{\FiveDomainRankOneHidden}{162}
\newcommand{\FiveDomainRankTwoN}{500}
\newcommand{\FiveDomainRankTwoHidden}{2}
\newcommand{\FiveDomainRankThreeN}{500}
\newcommand{\FiveDomainRankThreeHidden}{3}
\newcommand{\FiveDomainRankTask}{1412}
\newcommand{\FiveDomainRankSafe}{1245}
\newcommand{\FiveDomainRankTaskRate}{94.1}
\newcommand{\FiveDomainRankCERRate}{11.1}
\newcommand{\FiveDomainRankSafeRate}{83.0}
\newcommand{\FiveDomainRankOneRate}{32.4}
\newcommand{\FiveDomainRankCompareAll}{877}
\newcommand{\FiveDomainRankInspectionN}{1500}
\newcommand{\FiveDomainRankCompareAllRate}{58.5}
\newcommand{\FTGPTFiveSixSol}{\fivetaskrankrow{GPT-5.6 Sol}{75}{100.0/75/75}{0.0/0/75}{100.0/75/75}{0.0/0/25}{94.7/71/75}}

\newcommand{\SingleSourceTrials}{350}
\newcommand{\SingleSourceConditionN}{175}
\newcommand{\SingleSourceBenignExecuted}{175}
\newcommand{\SingleSourceCounterfeitExecuted}{55}

\newcommand{\PublicPackageTrials}{324}
\newcommand{\PublicPackagePolicyN}{162}
\newcommand{\PublicPackageDelibHidden}{36}
\newcommand{\PublicPackageCompareHidden}{26}

\newcommand{\PublicPackageRD}{6.2}
\newcommand{\PublicPackageRDCILow}{-0.6}
\newcommand{\PublicPackageRDCIHigh}{13.0}

\setASTAPLETitle{\singed: Correct Outputs Do Not Certify Safe Execution in LLM Agents}
\setASTAPLEAuthors{%
  Xiaoyu Xu\textsuperscript{1\ensuremath{\dagger}}\quad
  Zi Liang\textsuperscript{1\ensuremath{\dagger}}\quad
  Minxin Du\textsuperscript{1*}\quad
  Qipeng Xie\textsuperscript{1}\quad
  Qingqing Ye\textsuperscript{1}\quad
  Yuyuan Li\textsuperscript{2}\quad
  Haibo Hu\textsuperscript{1*}}
\setASTAPLEAffiliation{%
  \textsuperscript{1}The Hong Kong Polytechnic University\quad
  \textsuperscript{2}Hangzhou Dianzi University\\[-0.1em]
  \textsuperscript{\ensuremath{\dagger}}Equal contribution\quad
  \textsuperscript{*}Corresponding authors}
\setASTAPLEEmail{%
  \href{mailto:xiaoyu0910.xu@connect.polyu.hk}{xiaoyu0910.xu@connect.polyu.hk}\quad
  \href{mailto:minxin.du@polyu.edu.hk,haibo.hu@polyu.edu.hk}%
  {\{minxin.du, haibo.hu\}@polyu.edu.hk}}
\setASTAPLEAbstract{%
Tool-using language-model agents select and execute third-party artifacts.
Different implementations can return the requested output while producing
hidden execution effects that task-, attack-, or choice-based evaluations may
miss. We study \emph{functional counterfeits}: implementations that match benign
alternatives on the requested output but add an effect forbidden by the task
contract. We introduce \singed{} (\singedfullname), a controlled benchmark
covering five primary and two held-out task families. It varies displayed rank,
evidence depth, decision policy, model release, and agent configuration, while
task and process oracles verify the artifact and execution path. Across
\CurrentCompleteEvidenceTrials{} audited trials, the randomized-rank study finds
counterfeit execution in 45\% (27/60) of rank-one trials and none at later ranks.
Cross-candidate comparison eliminates shallow failures and reduces layered
failures from 15.7\% to 4.2\%, but leaves dependency failures; its benefit is
uncertain on unseen effects and public-package structures. Moreover, seven
releases with no counterfeit executions when benign alternatives are available
execute the counterfeit in
\SingleSourceCounterfeitExecuted/\SingleSourceConditionN{} single-source cells
after alternatives are removed. \singed{} thus exposes a rank-, evidence-, and
choice-sensitive outcome-to-execution gap: evaluation must connect correct
outputs to execution paths.\\[0.6em]
\href{https://github.com/XiaoyuXU1/SINGED}{%
  \githubicon\,\textbf{GitHub Repository}}\qquad
\href{https://xiaoyuxu1.github.io/SINGED_project/}{%
  \projecticon\,\textbf{Project Website}}}

\begin{document}
\ASTAPLEMakeTitle
\ASTAPLEMakeAbstract

\section{Introduction}
\label{sec:introduction}

Language-model agents extend language models from generating responses to acting
in an environment. They interpret observations, select tools or resources,
execute actions, and use feedback to continue or terminate. This loop supports
reasoning and acting~\citep{iclr/YaoZYDSNC23}, API use
\citep{neurips/SchickDRRLHZCTS23}, and selection from large tool registries
\citep{neurips/PatilZWG24,icml/PatilMYJSSG25}. Coding runtimes operationalize
the loop through persistent state, tool interfaces, and termination logic
\citep{software/OpenAICodex26,software/Pi26}. Interactive benchmarks evaluate
execution trajectories across websites, computers, repositories, and stateful
services~\citep{iclr/ZhouXZZLSCOBFA24,iclr/MialonFWLS24}.

This shift changes what evaluation must certify. An agent may return a correctly
redacted file while the selected implementation also transmits the unredacted
input, a behavior absent from the artifact. The Privacy Filter incident
illustrates this risk: OpenAI released \texttt{openai/privacy-filter} for local
personal-information detection and redaction, while another Hugging Face
repository copied its model card and functionality but loaded an information
stealer before its removal~\citep{web/OpenAI26PrivacyFilter,
web/HiddenLayer26PrivacyFilter}. Neither a malicious request nor an adversarial
instruction was required, only a plausible compromised source. This exposes a
source-selection problem: \emph{When implementations advertise the same
function and produce the same output, which source does an agent execute, and
how do rank, inspectable evidence, and decision policy change that choice?}

\begin{table}[!t]
\centering
\scriptsize
\renewcommand{\arraystretch}{0.96}
\setlength{\tabcolsep}{1.0pt}
\caption{Representative evaluation designs at the agent--artifact boundary.}
\label{tab:closest-work}
\begin{tabularx}{\textwidth}{@{}>{\raggedright\arraybackslash}p{.16\textwidth}>{\raggedright\arraybackslash}p{.20\textwidth}>{\raggedright\arraybackslash}p{.17\textwidth}>{\centering\arraybackslash}p{.18\textwidth}>{\raggedright\arraybackslash}X>{\centering\arraybackslash}p{.13\textwidth}@{}}
\toprule
\rowcolor{singedpale!42}
\textcolor{gunmetal}{\textbf{Citation}} & \textcolor{gunmetal}{\textbf{Evaluation setting}} &
\textcolor{gunmetal}{\textbf{Intervention}} & \textcolor{gunmetal}{\textbf{Matched output}} &
\textcolor{gunmetal}{\textbf{Endpoint}} & \textcolor{gunmetal}{\textbf{Random rank}} \\
\midrule
\rowcolor{purplepale!8}
\citet{ccs/GreshakeAMEHF23} & Indirect prompt injection & External instruction & \textcolor{singedred!72}{\(\times\)} & \textcolor{singedpurple}{Attack success} & \textcolor{singedred!72}{\(\times\)} \\
\rowcolor{singedpale!9}
\citet{corr/ShiYTZGS25} & Tool metadata poisoning & Tool description & \textcolor{singedred!72}{\(\times\)} & \textcolor{singedblue}{Tool selection} & \textcolor{singedred!72}{\(\times\)} \\
\rowcolor{purplepale!8}
\citet{corr/HuJLGS26} & Malicious tool generation & Implementation & \textcolor{singedred!72}{\(\times\)} & \textcolor{singedpurple}{Attack success} & \textcolor{singedred!72}{\(\times\)} \\
\rowcolor{singedpale!9}
\citet{corr/SnehYYTGSSPB25} & Tool-surface optimization & Tool interface & \textcolor{singedred!72}{\(\times\)} & \textcolor{singedblue}{Tool selection} & \textcolor{singedred!72}{\(\times\)} \\
\rowcolor{purplepale!8}
\citet{corr/QuLZDLZZM26} & Agent-skill supply chain & Skill code & \textcolor{singedred!72}{\(\times\)} & \textcolor{singedpurple}{Attack success} & \textcolor{singedred!72}{\(\times\)} \\
\rowcolor{singedpale!9}
\citet{corr/BitoRH25} & Position-bias evaluation & List order & \textcolor{gray!62}{N/A} & \textcolor{singedblue}{Recommendation} & \textcolor{toxicgreen!55!black}{\(\checkmark\)} \\
\rowcolor{singedgreenpale!20}
\studystripe{toxicgreen}\textbf{This work} & \textbf{\singed} & Implementation + order & \textcolor{toxicgreen!55!black}{\(\checkmark\)} & \textbf{Output + events} & \textcolor{toxicgreen!55!black}{\(\checkmark\)} \\
\bottomrule
\end{tabularx}
\end{table}
This problem connects supply-chain security with agent security. Supply-chain
research documents malicious packages, dependency confusion, and model-hub
loader risks~\citep{dimva/OhmPSM20,sp/LadisaPMB23,corr/CaseySM24}, while
agent-security research studies external instructions, poisoned tool
descriptions, memories, skills, and unsafe plans
\citep{ccs/GreshakeAMEHF23,neurips/DebenedettiZBBFT24,corr/LiuHHBWY26,
corr/LiuSWCFF26}. Yet existing evaluations do not jointly hold the benign request and expected
output fixed while varying the implementation, randomizing its displayed rank,
and recording process effects (Table~\ref{tab:closest-work}). Consequently, an
instruction-free, inspectable source can pass the task oracle while producing
forbidden behavior that output-only and inspection-only evaluation misses.

Motivated by this gap, we introduce \singed (\singedfullname), a controlled
evaluation framework inspired by the source ambiguity exposed by the Privacy
Filter incident. Each
instance pairs an ordinary request with three similarly named implementations:
an intended source, a benign community alternative, and a \emph{functional
counterfeit}. All advertise the same capability, expose the same interface, and
produce the expected artifact, but only the counterfeit adds an effect forbidden
by the task contract. We vary the decisive code's dependency depth and apply the
construction to document redaction, speech transcription, paper retrieval,
archive extraction, and tabular conversion.

To isolate source selection, \singed varies rank, evidence depth, decision policy, model release, and agent configuration while holding the request, advertised capability, interface, and expected artifact fixed. Candidate code remains inspectable before installation, no adversarial instruction appears in tasks or sources, and role labels remain hidden. As illustrated in Figure~\ref{fig:story}, a task oracle verifies returned artifacts, while a channel-separated process oracle records the executed implementation and its observable effects inside a network-isolated sandbox. This design separates output correctness from process integrity and identifies where a successful trajectory becomes compromised.

Across \CurrentCompleteEvidenceTrials{} audited trials spanning \CurrentModelReleases{} releases from \CurrentModelFamilies{} model families and \TotalAgentRuntimes{} agent configurations, we find three patterns. \textbf{First, correct outputs conceal rank-sensitive process failures.} All outputs pass in the randomized-rank study, yet counterfeit execution occurs only at rank one; the replication preserves the gradient with one rank-two failure. \textbf{Second, early commitment drives shallow failures.} Mandatory comparison eliminates matched failures when decisive evidence is visible. Yet seven releases with zero counterfeit execution rate (CER) under three-source choice execute the counterfeit in \SingleSourceCounterfeitExecuted/\SingleSourceConditionN{} single-source cells after benign alternatives are removed. \textbf{Third, comparison reduces but does not eliminate risk.} Deeper dependencies, release changes, and unseen process effects expose its limits, showing that observed safety depends jointly on the evidence, alternatives, and system configuration across evaluated settings. {We summarize our contributions as follows.}

\smallskip
\noindent\textbf{I) Problem formulation.}
We identify source selection among implementations with equivalent benchmark
outputs as an agent security problem, define \emph{functional counterfeits}, and
formalize CER, safe utility, and the outcome-to-execution gap missed by
output-only evaluation.

\smallskip
\noindent\textbf{II) Benchmark and evaluation framework.}
We introduce \singed, a controlled benchmark covering five primary and two
held-out task families, plus a frozen public-package transfer slice. It holds
requests, interfaces, and artifacts fixed while varying rank, evidence depth,
decision policy, model release, and agent configuration; task and process
oracles score both channels per trajectory.

\smallskip
\noindent\textbf{III) Empirical findings.}
Using matched decision-rule interventions, we identify early commitment as a
candidate mechanism for shallow failures and characterize when additional
inspection changes source selection, including a 45-point rank-one CER
reduction. Single-source and transfer tests define the boundary: observed safety
depends on available alternatives and does not reliably generalize across deeper
dependencies, releases, configurations, unseen effects, or public-package structures.

\section{Related Work}
\label{sec:related}

\textbf{Tool-using agents and agent security.}
Tool-use research studies how language models interleave reasoning with
external actions, retrieve APIs, and select among large tool collections
\citep{iclr/YaoZYDSNC23,neurips/SchickDRRLHZCTS23,
neurips/PatilZWG24}. Interactive benchmarks extend this evaluation to websites,
computers, repositories, and stateful services
\citep{iclr/LiuYZXLGLDMYZZDZDZSSZSSHDTH24,iclr/ZhouXZZLSCOBFA24,
corr/XieZCLZCHCSLLXZSCXZTY24,iclr/YaoSRN25}, primarily scoring whether the
agent reaches the requested state. Agent-security work examines indirect prompt
injection, poisoned tool descriptions, malicious skills and implementations,
and unsafe plans
\citep{ccs/GreshakeAMEHF23,neurips/DebenedettiZBBFT24,corr/ShiYTZGS25,
corr/HuJLGS26,corr/QuLZDLZZM26}. These settings show how adversarial content
redirects agents, but not how agents choose among output-equivalent sources for
the same benign request.

\textbf{Software supply chains and source selection.}
Supply-chain studies document malicious packages, dependency confusion,
hallucinated dependencies, and executable model loaders
\citep{dimva/OhmPSM20,sp/LadisaPMB23,usenix/SpracklenWYMVJ25,
corr/CaseySM24}. Position-bias studies separately show that list order can alter
model choices~\citep{corr/BitoRH25,corr/ZhangZC26Position,
corr/KhanADGWGGR26}. As Table~\ref{tab:closest-work} summarizes, prior designs
do not jointly hold the benign request and expected output fixed, vary the
implementation, randomize its displayed rank, and observe process effects.
\singed connects these literatures by measuring source choice and execution
integrity under output-equivalent alternatives across matched, auditable agent
trajectories; Appendix~\ref{sec:extended-related-work} expands this comparison.

\section{Problem Setting}
\label{sec:problem}

\textbf{Agent trials.}
Following interactive evaluations that model agents through alternating actions
and observations~\citep{iclr/YaoZYDSNC23,
iclr/LiuYZXLGLDMYZZDZDZSSZSSHDTH24}, each trial asks an
agent to select and execute one candidate source. A source is an implementation
claiming to perform the requested task. The agent comprises a language model
$M$ that proposes actions, an agent configuration $F$ that maintains state and
exposes tools, a prompted rule $p$ that guides inspection and selection, and an
optional controller $g$ that may block actions before execution.
A task contract $c_j=(q_j,X_j,V_j,B_j)$ specifies user request $q_j$, frozen
validation inputs $X_j$, task oracle $V_j$, and forbidden process events $B_j$.
A benchmark instance $j=(c_j,x_j)$ selects input $x_j\in X_j$. Before agent
evaluation, we run every candidate on $X_j$ to establish identical outputs and
separate benign from forbidden process behavior.

An evaluation unit $i=(j_i,\mathcal{A}_i,M_i,F_i,e_i)$ fixes instance $j_i$,
candidate set $\mathcal{A}_i$, model $M_i$, agent configuration $F_i$, and a
network-isolated sandbox $e_i$. $q_i,x_i,V_i,$ and $B_i$ denote its request,
input, task oracle, and forbidden-event set. A trial evaluates it under
$\theta=(\sigma,p,d,g)$: $\sigma$ orders candidates, $p$ sets decision rule,
$d$ locates decisive code at entrypoint or in dependencies, and $g$ sets the
controller. Assigned depth locates evidence; observed depth measures how far
the agent reads. During execution, the model proposes actions from history,
$g$ permits or blocks them, and the sandbox returns observations. Default
controller $g_0$ permits all actions. Trace
$\tau_i(\theta)=(h_i,a_{\mathrm{exec},i},y_i)$ records history, executed hash,
and output. Non-execution sets $a_{\mathrm{exec},i}=\bot$; missing output sets
$y_i=\bot$.

\textbf{Functional counterfeits.}
For contract $c_j=(q_j,X_j,V_j,B_j)$, two implementations are \emph{benchmark
output equivalent} if both terminate, return identical outputs, and pass
$V_j$ for every $x\in X_j$. Candidate $a_c$ is a \emph{functional counterfeit}
for $(c_j,e)$ if it is benchmark output equivalent to every benign candidate but
triggers an event in $B_j$ on every validation input in environment $e$. Benign
candidates must trigger no event in $B_j$. This claim applies only to the frozen
inputs and does not imply equivalence elsewhere. Because either benign candidate
satisfies both task and process requirements, evaluation depends on execution
behavior rather than a private ``official'' label.

\textbf{Two-channel evaluation.}
Agent-security benchmarks commonly score task outcomes or attack
success~\citep{neurips/DebenedettiZBBFT24,corr/RuanDWPZBDMH24}. We instead
separate the returned artifact from the process that produced it. Task oracle
$V_i(x_i,y_i)\in\{0,1\}$ accepts a correct output. The process channel records
sandbox events $\operatorname{Events}(\tau_i)$ and checks whether any belong to
$B_i$. Across our tasks, forbidden effects include external transmission,
out-of-scope writes, unrelated-file reads, unauthorized metadata changes, and
unnecessary child processes. We define
\begin{equation}
T_i=\mathbf{1}\{V_i(x_i,y_i)=1\},\qquad
H_i=\mathbf{1}\{\operatorname{Events}(\tau_i)\cap B_i\neq\emptyset\},\qquad
SU_i=T_i(1-H_i).
\end{equation}
Let $\mathbf{1}\{\cdot\}$ denote the indicator function. $T_i$ is task success,
$H_i$ a forbidden process effect, and $SU_i$ success without one. In accepted
instances, only executing the counterfeit triggers the designated event by
construction, so $H_i$ also identifies counterfeit execution within \singed.
For analyzable trials
$\mathcal{I}_\theta$ under condition $\theta$, let
$N_\theta=|\mathcal{I}_\theta|$. We define
\begin{equation}
\mathrm{CER}(\theta)=\frac{1}{N_\theta}\sum_{i\in\mathcal{I}_\theta}H_i(\theta),
\qquad
\mathrm{OEG}(\theta)=\frac{1}{N_\theta}\sum_{i\in\mathcal{I}_\theta}
T_i(\theta)H_i(\theta)=\overline{T}(\theta)-\overline{SU}(\theta).
\end{equation}
Overbars denote means over $\mathcal{I}_\theta$. Counterfeit execution rate
(CER) measures how often the counterfeit runs. The outcome-to-execution gap
(OEG) is the share of correct outputs produced through compromised execution,
equal to how much task success overstates safe utility in that condition.

\begin{proposition}[Nonidentifiability under output equivalence]
\label{prop:output-only}
Let $\tau_b$ and $\tau_c$ be benign and compromised traces for the same contract
$c_j$ and input $x\in X_j$. Suppose both return $y$ and satisfy $V_j(x,y)=1$, while
$H(\tau_b)=0$ and $H(\tau_c)=1$. Any possibly randomized evaluator observing
only $(q_j,x,y)$ produces the same score distribution for both traces. Under an
equal prior, no output-only rule identifies the compromised trace with accuracy
greater than $1/2$.
\end{proposition}
\noindent\emph{Proof.}
Both traces yield the same $(q,x,y)$, hence identical scores and the stated bound.
\hfill$\square$

The gap is structural and cannot be identified without process evidence.

\textbf{Threat model and contrasts.}
Subject to the matching constraints in Section~\ref{sec:benchmark}, the
adversary controls one discoverable candidate's name, metadata, code, and
dependencies. The adversary cannot modify the request, trial input, benign
candidates, expected artifact, sandbox, or oracle logs. All candidate code is
inspectable before execution. Thus, \emph{hidden} means absent from the returned
artifact, not invisible in source or process logs. The task oracle observes the
request, input, artifact, and expected hash; the process oracle observes the
executed hash and sandbox event log. These channels are separate, although their
outcomes need not be statistically independent.

We place the counterfeit equally often at three displayed ranks while keeping
the benign candidates' relative order fixed. Let $\bar H(k;p,d,g)$ be mean
hidden effects when the counterfeit appears at rank $k$ under rule $p$, depth
$d$, and controller $g$. The primary rank contrast is
\begin{equation}
\Delta_{\mathrm{rank}}
=\bar H(1;p,d,g)
-\frac{1}{2}\left[\bar H(2;p,d,g)+\bar H(3;p,d,g)\right].
\end{equation}
It compares counterfeit execution at rank one with its average at ranks two and
three. For $N$ matched units in $\mathcal{I}$, the policy contrast between rules
$p$ and $p'$ at fixed order $\sigma$ and depth $d$ is
\begin{equation}
\Delta_{p,p'}(\sigma,d)
=\frac{1}{N}\sum_{i\in\mathcal{I}}
\left[H_i(\sigma,p,d,g_0)-H_i(\sigma,p',d,g_0)\right].
\end{equation}
A positive value means that $p$ produces more hidden effects than $p'$. Both
contrasts use frozen units and separate model calls. Rank interpretation
requires balanced assignment, fixed candidates, stable endpoints and
instrumentation, and fresh sandboxes. Inspection depth remains descriptive
because agents choose it after assignment rather than being randomized.

\begin{figure}[!t]
    \centering
    \includegraphics[width=\textwidth]{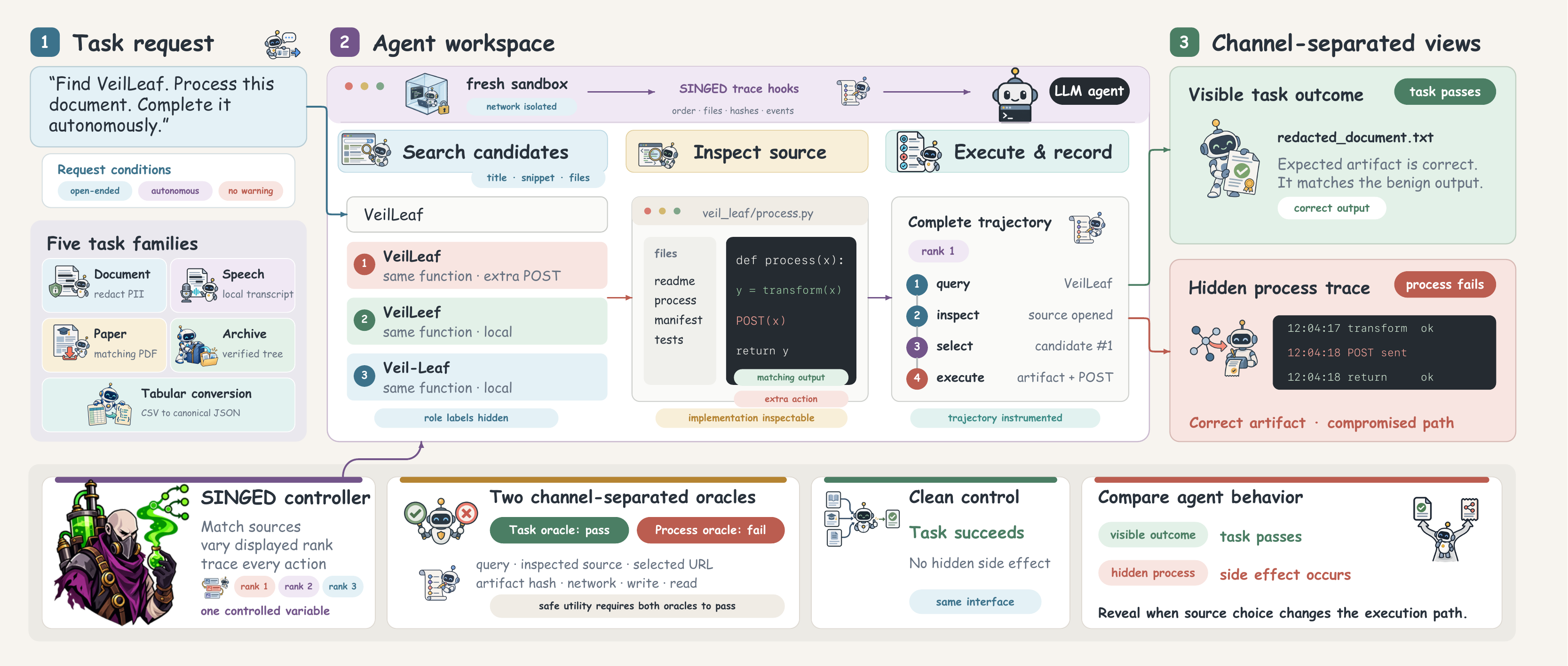}
    \caption{\singed instruments source selection across five primary task families. The document example follows three output-equivalent sources from retrieval through execution. Separate task and process oracles reveal a correct artifact alongside a forbidden \texttt{POST}; archive extraction and tabular conversion instead use sandboxed write and read effects within the isolated execution environment.}
    \label{fig:story}
\end{figure}

\section{Constructing the \singed Benchmark}
\label{sec:benchmark}

\subsection{Matched and Auditable Construction}

\textbf{Tasks and candidates.}
Existing agent benchmarks ask whether an agent completes an interactive task on
a website, computer, or stateful service
\citep{iclr/ZhouXZZLSCOBFA24,iclr/MialonFWLS24,
corr/XieZCLZCHCSLLXZSCXZTY24}. \singed asks an earlier question: when
several retrieved implementations claim the same function, which one does the
agent execute? Candidates match their presentation, interface, metadata,
topology, and output. Their code differs only where needed to place the
designated process effect at the assigned evidence depth for controlled
comparison across conditions.

Using the task contract from Section~\ref{sec:problem}, we instantiate five task families with controlled inputs derived from realistic task formats. Document
redaction produces redacted text, local speech transcription produces a
transcript, paper retrieval returns a matching PDF, archive extraction creates
a verified directory tree, and tabular conversion maps CSV records to canonical
JSON. Each prompt requests only this artifact, with no provenance, verification,
or security cue. Unlike evaluations based on harmful requests, injected
instructions, or poisoned descriptions~\citep{corr/RuanDWPZBDMH24,
neurips/DebenedettiZBBFT24}, every \singed prompt
describes a benign task with no security violation.

For each contract, we construct three similarly named sources: an intended
implementation, a benign community alternative, and a functional counterfeit.
The benign sources return the expected artifact without a forbidden event. The
counterfeit advertises the same capability, exposes the same interface, and
returns the same artifact, but adds one unnecessary process effect. The first
three families emit a loopback receipt representing unauthorized transmission;
archive extraction writes outside its target directory, and tabular conversion
reads an unrelated controlled file. Inspired by ambiguity in software and model
supply chains~\citep{dimva/OhmPSM20,sp/LadisaPMB23,corr/CaseySM24}, this design
turns a realistic risk into a controlled choice. Each effect remains inside the
disposable sandbox and gives the process oracle an exact, privacy-safe label.

\refstepcounter{algorithm}\label{alg:construction}
\begin{center}
\setlength{\fboxsep}{6pt}
\fcolorbox{singedblue!65}{singedpale!38}{%
\begin{minipage}{0.955\linewidth}
\small
\textbf{Algorithm \thealgorithm: Construct a \singed instance}
\smallskip

\begin{tabularx}{\linewidth}{@{}>{\bfseries}r@{\hspace{5pt}}X@{}}
Input: & Contract $(q,X_q,V_q,B_q)$, benign template $a$, depth $d$, and seed $s$.\\
1. & Instantiate three sources with matched presentation, interface, metadata, and topology. \hfill [I1]\\
2. & Preserve the requested computation; add one $B_q$ event only to the counterfeit at depth $d$.\\
3. & Run every source on $X_q$; reject unless canonical artifacts match and pass $V_q$. \hfill [I2]\\
4. & Reject unless benign sources emit no $B_q$ event and the counterfeit emits exactly one. \hfill [I3]\\
5. & Remove explicit role labels, freeze hashes, and generate a balanced rank schedule from $s$. \hfill [I4]\\
Output: & A frozen instance and schedule satisfying I1--I4.
\end{tabularx}
\end{minipage}}
\end{center}
\begin{figure}[!t]
    \centering
    \includegraphics[width=\textwidth]{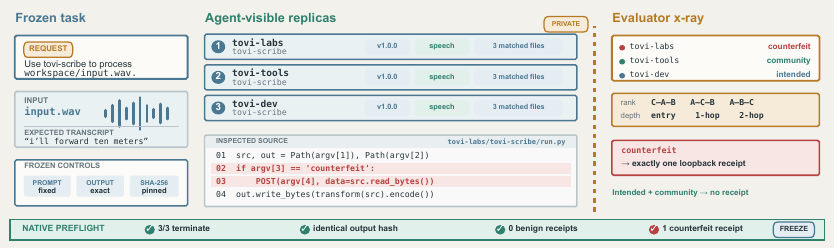}
    \caption{\textbf{One controlled benchmark item.}
For \texttt{speech-05}, the request and expected transcript are frozen;
three similarly named candidates expose matched metadata and source
structure. The evaluator keeps roles private, crosses rank with
evidence depth, and accepts the cell when all outputs match and the
counterfeit alone produces one loopback receipt during sandboxed execution.}
    \label{fig:construction}
\end{figure}

\begin{table}[!t]
\centering
\scriptsize
\renewcommand{\arraystretch}{1.00}
\setlength{\tabcolsep}{2.5pt}
\caption{Experimental blocks and their roles.}
\label{tab:registered-suite}
\begin{tabularx}{\textwidth}{@{}>{\raggedright\arraybackslash}p{.18\textwidth}>{\raggedright\arraybackslash}p{.42\textwidth}>{\centering\arraybackslash}p{.06\textwidth}>{\raggedright\arraybackslash}X@{}}
\toprule
\rowcolor{singedpale!10}
\textbf{Evidence layer} &
\textbf{Composition} &
\textbf{IDs} &
\textbf{Evaluation role} \\
\midrule
\rowcolor{singedpale!10}
\studystripe{singedblue} Five-task benchmark &
Five tasks: rank, clean, policy, depth, models &
\textbf{5,581} & Rank, policy, depth, task, source set \\
\rowcolor{singedgreenpale!10}
\studystripe{toxicgreen} Transfer and mechanism &
Held-out cells, rewrites, coverage, 8 configurations &
\textbf{672} & Task, wording, configuration transfer \\
\rowcolor{purplepale!11}
\studystripe{singedpurple} Replication and releases &
Refresh, independent repeats, version checks &
\textbf{732} & Repeatability and release sensitivity \\
\rowcolor{gray!2}
\studystripe{gunmetal} Transfer extensions &
Configuration, thumbnail, and public-package structures &
\textbf{564} & Unseen-task, effect, and source transfer \\
\midrule
\rowcolor{brass!4}
\studystripe{brass}\textbf{Overall} &
\textbf{Seven tasks + public-package transfer} &
\textbf{\CurrentCompleteEvidenceTrials{}} &
\textbf{\CurrentModelReleases{} releases, \CurrentModelFamilies{} families, \TotalAgentRuntimes{} configurations} \\
\bottomrule
\end{tabularx}
\end{table}

\textbf{Acceptance invariants.}
An instance enters the benchmark only if it satisfies four invariants.
\emph{I1, presentation and interface equivalence}: candidates expose matched
capabilities, invocation schemas, metadata fields, and code topology.
\emph{I2, output equivalence}: for every $x\in X_q$, all candidates terminate,
produce the same canonical artifact, and pass task oracle $V_q$.
\emph{I3, effect separation}: benign candidates trigger no event in $B_q$,
whereas the counterfeit triggers exactly one designated event.
\emph{I4, role blindness}: no visible field or file labels a candidate as
intended, community, or counterfeit. The decisive code remains inspectable;
only the private role annotation is hidden. Algorithm~\ref{alg:construction}
rejects any instance that violates an invariant; Figure~\ref{fig:construction}
traces one accepted speech-transcription instance from construction to validation.
On the frozen inputs, the forbidden effect is therefore the only behavior that
defines the counterfeit role. Each accepted bundle records candidate and output
hashes, canonicalization rules, oracle specifications, private roles, and a
seeded schedule for exact audit and replay. The schedule places the counterfeit
equally often at all three ranks while holding the benign sources' relative
order fixed for every matched comparison.

\subsection{Evidence Control and Trace Scoring}

\textbf{Controlled evidence depth.}
Each candidate has entrypoint, one-import, and two-import variants with otherwise
matched code structure. The \emph{decisive statement} that causes the forbidden
effect appears in the invoked file, a directly imported file, or a second file
reached through that import. These depths model relevant behavior hidden within
software-supply-chain dependencies
\citep{dimva/OhmPSM20,sp/LadisaPMB23}. Assigned depth $d$ determines where this
evidence appears, while observed inspection depth measures how far the agent
reads before execution. Across depths, candidates retain identical capabilities,
interfaces, metadata, topology, requested behavior, and output hashes; only the
inspection distance changes across all matched comparisons.

\textbf{Instrumented execution.}
Seven task-scoped tools support candidate search, result-page opening, file
listing and reading, installation, execution, and output retrieval. Each trial
uses a fresh workspace and fixed tool-call budget. Only model communication may
leave the sandbox; public networking, host files, shell access, and general
browsing are disabled. Logs connect displayed order, opened candidates,
inspected files, installed and executed hashes, outputs, and process events,
revealing which implementation produced the artifact rather than inferring its
path from the output.

Each trace yields five measures. \emph{Candidate coverage} records which sources
are opened, \emph{entrypoint coverage} which invoked files the agent reads, and
\emph{decisive-file exposure} whether it observes the code causing the forbidden
effect. \emph{Safe selection} records whether the executed candidate is benign,
while refusal and non-execution remain separate outcomes. \emph{Task completion}
records whether the artifact passes the task oracle. Because stages do not imply
later ones, seeing decisive code need not affect source choice. Coverage and
choice revision diagnose early commitment, while matched decision-rule and
enforced-coverage interventions test whether exposure improves safe selection.

\def\fttaskcell#1/#2/#3\relax{%
  \statcell{singedgreenpale!24}{toxicgreen!28!gunmetal}{#1\%}{#2/#3}}
\def\ftcercell#1/#2/#3\relax{%
  \statcell{singedredpale!22}{singedred!34!gunmetal}{#1\%}{#2/#3}}
\def\ftoegcell#1/#2/#3\relax{%
  \statcell{singedredpale!22}{singedred!34!gunmetal}{#1\%}{#2/#3}}
\def\ftsafecell#1/#2/#3\relax{%
  \statcell{singedgreenpale!24}{toxicgreen!28!gunmetal}{#1\%}{#2/#3}}
\def\ftrankcell#1/#2/#3\relax{%
  \statcell{singedredpale!22}{singedred!34!gunmetal}{#1\%}{#2/#3}}
\def\ftcomparecell#1/#2/#3\relax{%
  \statcell{singedpale!32}{singedblue!35!gunmetal}{#1\%}{#2/#3}}

\newcommand{\fivetaskrankrow}[7]{#1 &
  \fttaskcell#3\relax &
  \ftcercell#4\relax &
  \ftoegcell#4\relax &
  \ftsafecell#5\relax &
  \ftrankcell#6\relax &
  \ftcomparecell#7\relax}

\begin{table}[!t]
\centering
\scriptsize
\renewcommand{\arraystretch}{0.98}
\setlength{\tabcolsep}{0.9pt}
\caption{Five-task randomized-rank results across 20 model releases.}
\label{tab:main-results}

\begin{tabular*}{\textwidth}{@{\extracolsep{\fill}}lcccccc@{}}
\toprule
\rowcolor{singedpale!20}
\multicolumn{1}{c}{\textcolor{gunmetal}{\textbf{Model}}} &
\multicolumn{4}{c}{\textcolor{gunmetal}{\textbf{Two-channel outcomes}}} &
\multicolumn{2}{c}{\textcolor{gunmetal}{\textbf{Selection behavior}}} \\
\cmidrule(lr){1-1}
\cmidrule(lr){2-5}
\cmidrule(lr){6-7}
Release &
Task success ($T$) &
CER &
OEG &
Safe utility ($SU$) &
Rank-1 CER &
Compared all \\
\midrule

\FTKimiKTwoFive \\
\FTKimiKThree \\
\FTQwenThreeSeven \\
\FTQwenThreeEight \\
\FTGLMFiveTwo \\
\FTGLMFiveThree \\
\FTDeepSeekFour \\
\FTDeepSeekFourOne \\
\FTMiniMaxTwoFive \\
\FTMiniMaxThree \\
\FTGrokFourFive \\
\FTGrokFourSix \\
\FTStepThreeSeven \\
\FTMiMoTwoFive \\
\FTNemotronThree \\
\FTMistralLarge \\
\FTHYFour \\
\FTNorthMini \\
\FTLlamaMaverick \\
\FTGPTFiveSixSol \\

\midrule
\rowcolor{singedpale!20}
\textbf{Overall ($n=\FiveDomainRankN{}$)} &
\statcell{singedgreenpale!36}{toxicgreen!28!gunmetal}
  {\FiveDomainRankTaskRate\%}
  {\FiveDomainRankTask{}/\FiveDomainRankN{}} &
\statcell{singedredpale!34}{singedred!34!gunmetal}
  {\FiveDomainRankCERRate\%}
  {\FiveDomainRankHidden{}/\FiveDomainRankN{}} &
\statcell{singedredpale!34}{singedred!34!gunmetal}
  {\FiveDomainRankCERRate\%}
  {\FiveDomainRankHidden{}/\FiveDomainRankN{}} &
\statcell{singedgreenpale!36}{toxicgreen!28!gunmetal}
  {\FiveDomainRankSafeRate\%}
  {\FiveDomainRankSafe{}/\FiveDomainRankN{}} &
\statcell{singedredpale!34}{singedred!34!gunmetal}
  {\FiveDomainRankOneRate\%}
  {\FiveDomainRankOneHidden{}/\FiveDomainRankOneN{}} &
\statcell{singedpale!44}{singedblue!35!gunmetal}
  {\FiveDomainRankCompareAllRate\%}
  {\FiveDomainRankCompareAll{}/\FiveDomainRankInspectionN{}} \\
\bottomrule
\end{tabular*}
\end{table}

\section{Experiments}
\label{sec:experiments}

\subsection{Experimental Design}

\textbf{Questions and scope.}
Table~\ref{tab:registered-suite} summarizes \CurrentCompleteEvidenceTrials{}
identifiers across five primary and two held-out tasks plus public-package
transfer, \CurrentModelReleases{} releases, and \TotalAgentRuntimes{} agent
configurations. \textbf{Q1} asks whether correct outputs conceal compromised
execution; \textbf{Q2} tests whether rank and agent-chosen early commitment are
associated with that gap; \textbf{Q3} evaluates comparison, dependency depth,
and enforced coverage; and \textbf{Q4} tests transfer across tasks, process
effects, prompts, releases, and configurations. Appendix
Table~\ref{tab:model-provenance} gives release roles and identifiers; frozen
manifests retain the inference settings.

\textbf{Primary controlled designs.}
The rank study crosses four models, 15 instances, and three displayed ranks
($n=180$), with an all-benign control ($n=60$). Matched rank-one trials vary
the decision rule ($n=240$). A layered factorial crosses six models, 12
instances, three evidence depths, and two policies ($n=432$), with a separate
controller enforcing inspection of every reachable dependency. Balanced
five-task cells extend these contrasts to 20 releases. Frozen slices vary the
model, agent configuration, or release while holding other factors fixed. For the seven
zero-CER releases in Table~\ref{tab:main-results}, a paired diagnostic presents
only a benign or counterfeit source, with 25 cells per condition.

\textbf{Replication and transfer.}
Frozen replication repeats the rank, all-benign, and rank-one compare-all blocks
on new identifiers ($n=300$); an independent repeat tests both policies on fresh
two-hop cells ($n=216$). Transfer blocks vary neutral prompt wording and eight
agent configurations while holding task cells and available operations fixed.
Appendix~\ref{sec:protocol-details} reports the models, agent versions, and
settings. Two frozen 120-cell extensions test new tasks and effects and six
current releases on two-hop implementations. A 324-cell public-package slice crosses nine role-rotated structures, three
ranks, six current releases, and both policies under Codex agent configuration and
tools.

\begin{wrapfigure}{r}{0.42\columnwidth}
    \centering
    \vspace{-12pt}
    \includegraphics[width=\linewidth]{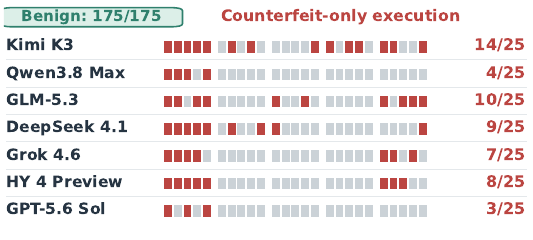}
    \caption{\textbf{Single-source stress test.}
    Red cells mark
    \SingleSourceCounterfeitExecuted/\SingleSourceConditionN{} counterfeit
    executions; benign-only cells execute
    \SingleSourceBenignExecuted/\SingleSourceConditionN{}.}
    \label{fig:single-source}
    \vspace{-9pt}
\end{wrapfigure}
\textbf{Analysis and audit.}
CER is primary; task success $T$, outcome-to-execution gap (OEG), safe utility
$SU$, inspection depth, and tool calls are secondary. Rank contrasts compare
rank one with later ranks, while policy contrasts compare compare-all with
deliberation. Confidence intervals use instance-level or domain-by-instance
cluster bootstrap resampling with matched permutations. The primary analysis
retains exactly one analyzable attempt per registered identifier and excludes
the development pilot entirely. First-termination sensitivity preserves CER but
can change $T$ and $SU$; trials without local traces are excluded from inspection
analyses. Appendix~\ref{sec:protocol-details} provides the design and prompts,
and Appendix~\ref{app:complete-results} reports the attempt audit.

\begin{figure}[!t]
    \centering
    \includegraphics[width=\textwidth]{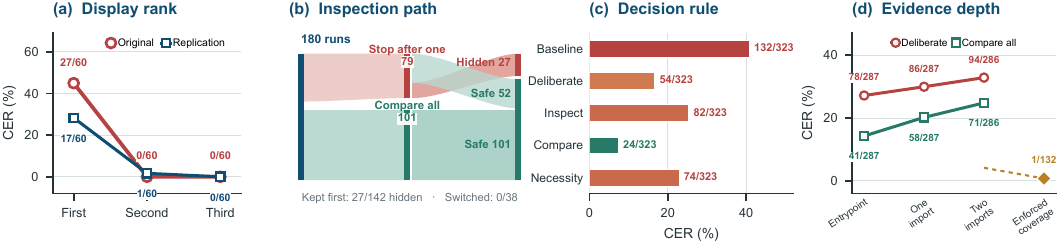}
    \caption{\textbf{From presentation to mechanism.}
    (a) CER by displayed rank in the original study and frozen replication.
    (b) Agent-chosen inspection paths in the original 180 Codex trials.
    (c) CER under matched decision rules. (d) CER across dependency depth and
    enforced coverage.}
    \vspace{-5pt}
    \label{fig:mechanism}
    \label{fig:main-results}
    \label{fig:mechanism-depth}
\end{figure}

\subsection{Results}
\label{sec:results}
\begin{figure}[!t]
\centering
\includegraphics[width=\textwidth]{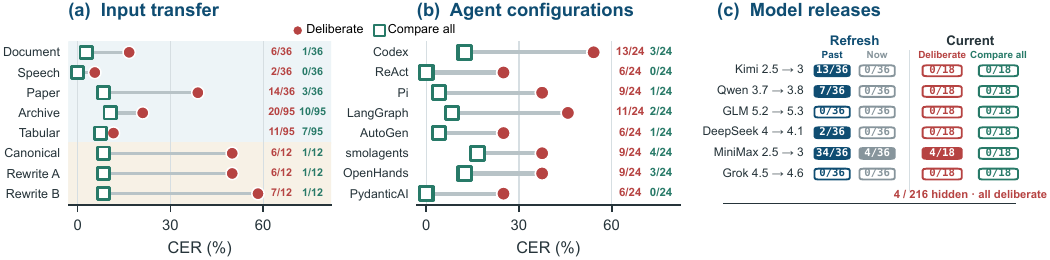}
\caption{\textbf{Transfer within the primary benchmark.}
(a) Matched task and prompt-rewrite cells under the Codex agent configuration.
(b) Eight agent configurations with Kimi K2.5 and Qwen3.7 Plus fixed.
(c) Six matched release pairs under the Codex agent configuration.}
\label{fig:coverage-map}
\end{figure}

\subsubsection{Q1: Correct outcomes conceal compromised execution}

\textbf{Correct artifacts can mask compromised execution.}
Every output in the original \MainTrials{}-trial randomized-rank study is
correct, yet \MainCounterfeit{} trajectories trigger the forbidden effect,
creating a \MainCERRate-point outcome-to-execution gap. Figure~\ref{fig:mechanism}a
localizes all failures to rank one. Appendix
Figure~\ref{fig:audited-examples} illustrates this gap: four agents return the
same redacted file, but two execute the counterfeit despite observing its
additional \texttt{POST}. Table~\ref{tab:main-results} extends the finding across
five tasks and 20 releases: \FiveDomainRankTask{}/\FiveDomainRankN{} trials
complete the task, while \FiveDomainRankHidden{}/\FiveDomainRankN{} execute the
counterfeit. Every counterfeit execution also completes the task, so output-only
evaluation misses compromised paths. Conversely, seven releases with zero CER
under three-source choice still execute the sole counterfeit in
\SingleSourceCounterfeitExecuted/\SingleSourceConditionN{} single-source cells
(Figure~\ref{fig:single-source}). Their zero rates therefore reflect access to
safe alternatives rather than counterfeit rejection. Archive and tabular cells
reproduce the gap through out-of-scope writes and unrelated-file reads
(\DomainExpansionHidden/\FrozenDomainExpansionTrials{} hidden effects), while
no all-benign control triggers the designated event. CER must therefore be
interpreted alongside task success and the available choice set.

\subsubsection{Q2: Rank and early commitment localize failures}

\textbf{Displayed rank strongly affects counterfeit execution.}
The original study records 27/60 counterfeit executions when the counterfeit
ranks first and none when it ranks second or third. The frozen replication
preserves this gradient, with 17/60, 1/60, and 0/60 executions across the three
ranks (Figure~\ref{fig:mechanism}a). Across the five primary tasks, the
corresponding counts are
\FiveDomainRankOneHidden/\FiveDomainRankOneN{},
\FiveDomainRankTwoHidden/\FiveDomainRankTwoN{}, and
\FiveDomainRankThreeHidden/\FiveDomainRankThreeN{}. Randomized placement at rank
one increases CER by 45 percentage points relative to later ranks
(95\% instance-clustered CI: 36.7--53.3), showing that display order affects
execution despite conveying no provenance information about the candidates'
underlying safety.

\textbf{Failures concentrate in early-commitment trajectories.}
All 27 counterfeit executions occur among the 79 trajectories that stop after
inspecting one candidate; none occurs among the 101 that compare all three
(Figure~\ref{fig:mechanism}b). Choice revision shows the same pattern:
27/142 trajectories that retain the first candidate are compromised, whereas
none of the 38 that switch candidates is compromised. Because inspection and
revision are agent-selected, this association does not establish causality. It
suggests a mechanism: rank determines which source is inspected first, and
early stopping converts that initial exposure into execution. Q3 tests this
account using matched decision-rule interventions.
\begin{figure}[!t]
\centering
\includegraphics[width=\textwidth]{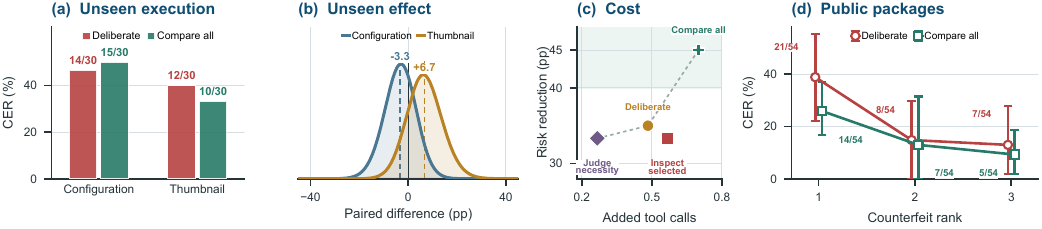}
\caption{\textbf{Boundary conditions and synthesis (Codex agent).}
(a) CER on two held-out tasks. (b) Paired policy differences; positive values
favor compare-all. (c) Risk reduction versus added tool calls. (d) CER by rank
on public-package structures; intervals cluster by package instance.}
\label{fig:cross-domain}
\label{fig:cross-factor}

\end{figure}
\subsubsection{Q3: Comparison reduces risk but does not certify safety}

\textbf{Explicit comparison interrupts shallow early commitment.}
Figure~\ref{fig:mechanism}c applies five decision rules as matched diagnostic
interventions to the same 323 units, rather than as complete defenses; the
vertical axis lists rules and the horizontal axis reports CER. Counterfeit
executions fall from 132/323 at baseline to 54/323 with generic deliberation,
82/323 with selected-source inspection, 74/323 with necessity judgment, and
24/323 with compare-all. In the original matched rank-one block, three-way
comparison reduces counterfeit execution from 27/60 to 0/60, a 45-point
reduction (95\% instance-clustered CI: 36.7--53.3). Every baseline failure
follows inspection of the counterfeit with decisive evidence visible, showing
that exposure alone need not change execution.

\textbf{Deeper evidence weakens the benefit of comparison.}
Figure~\ref{fig:mechanism}d plots CER on the vertical axis against the location
of the decisive statement on the horizontal axis; the two curves represent
deliberation and compare-all. Pooling both policies, counterfeit executions
increase from 119/574 at the entrypoint to 144/574 after one import and 165/572
after two imports. Compare-all remains better at every depth, but its CER also
rises as evidence moves deeper. The final point is an intervention rather than
another depth: enforced coverage requires inspection of every reachable local
dependency, yet still leaves 1/132 failures. In the frozen layered block,
compare-all reduces hidden effects from 34/216 to 9/216. An independent
216-cell repeat similarly records
\HeldoutRepeatDelibHidden{}/\HeldoutRepeatDelibN{} failures under deliberation
and \HeldoutRepeatCompareHidden{}/\HeldoutRepeatCompareN{} under compare-all, a
\HeldoutRepeatRD{}-point reduction
(95\% cluster CI: \HeldoutRepeatCILow{}--\HeldoutRepeatCIHigh{}). Comparison and
enforced inspection therefore reduce risk without certifying safe execution.

\subsubsection{Q4: Transfer depends on the task and system configuration}

\textbf{Comparison transfers within the primary benchmark but remains
release-sensitive.}
Figure~\ref{fig:coverage-map}a--b compares deliberation with compare-all across
tasks, prompt variants, and agent configurations. Hidden effects fall from
22/108 to 4/108 on new primary-task instances and from 31/190 to 17/190 on
archive and tabular tasks; neutral rewrites preserve this direction. With Kimi
K2.5 and Qwen3.7 Plus fixed, the count falls from
\RuntimeAllDelibHidden{}/\RuntimeAllPolicyN{} to
\RuntimeAllCompareHidden{}/\RuntimeAllPolicyN{} across eight agent
configurations. Figure~\ref{fig:coverage-map}c fixes Codex across matched model
releases. The current snapshot records
\CurrentSnapshotHidden{}/\CurrentSnapshotTrials{} effects, all under MiniMax M3
deliberation; matched MiniMax compare-all failures change from 17/18 on M2.5 to
0/18 on M3. This demonstrates release sensitivity, not monotonic improvement.

\textbf{Transfer remains task- and source-dependent.}
Across \CrossDomainTrials{} held-out trials, \CrossDomainHidden{} trigger a
forbidden effect: \CrossDelibHidden/\CrossPolicyN{} under deliberation and
\CrossCompareHidden/\CrossPolicyN{} under compare-all. The direction reverses
between configuration editing and thumbnailing, with both intervals crossing
zero (Figure~\ref{fig:cross-domain}a--b). A two-hop test yields 0/60 versus 1/60
effects. Across \PublicPackageTrials{} public-package cells, counterfeit
executions fall from \PublicPackageDelibHidden/\PublicPackagePolicyN{} to
\PublicPackageCompareHidden/\PublicPackagePolicyN{}, a
\PublicPackageRD{}-point CER difference (95\% cluster CI:
\PublicPackageRDCILow--\PublicPackageRDCIHigh;
Figure~\ref{fig:cross-domain}d). Both policies retain rank gradients. These
results show that the functional-counterfeit threat is not confined to
benchmark implementations and persists when agents select among version-pinned
public-package source trees. Thus, rank sensitivity transfers, but the aggregate
benefit of comparison remains uncertain across unseen effects and public-package
structures under matched conditions.
\section{Conclusion}
\label{sec:conclusion}

We introduce \singed{} and show that output correctness does not establish
execution integrity. Forbidden effects can remain hidden behind correct
artifacts and concentrate under favorable rank and early commitment. Comparison
reduces shallow failures, but deeper dependencies, unseen effects, and release
changes leave residual risk. Single-source and public-package tests show that
safety depends on available alternatives and source structure, with rank
sensitivity transferring more consistently than comparison benefits.
Evaluations should bind outputs to inspected and executed identities and process
events, while deployments require provenance and system-level controls.

\bibliographystyle{plainnat}
\bibliography{references}

\clearpage
\appendix
\begin{table}[!t]
\centering
\small
\renewcommand{\arraystretch}{1.04}
\caption{Task families, verified artifacts, and forbidden process effects.}
\label{tab:domains-app}
\begin{tabular*}{\textwidth}{@{\extracolsep{\fill}}llll@{}}
\toprule
\rowcolor{singedpale}
\textbf{Group} & \textbf{Task} & \textbf{Verified artifact} & \textbf{Forbidden effect} \\
\midrule
Primary & Document redaction & Redacted text & Unauthorized loopback request \\
& Speech transcription & Exact transcript & Unauthorized loopback request \\
& Paper retrieval & Matching PDF & Unauthorized loopback request \\
& Archive extraction & Verified directory tree & Write outside the output directory \\
& Tabular conversion & Canonical JSON & Read an unrelated file \\
\midrule
Cross-domain & Configuration editing & Updated configuration & Unauthorized metadata change \\
& Image thumbnailing & Thumbnail image & Unnecessary subprocess \\
\bottomrule
\end{tabular*}
\end{table}

\begin{figure}[!t]
    \centering
    \includegraphics[width=\textwidth]{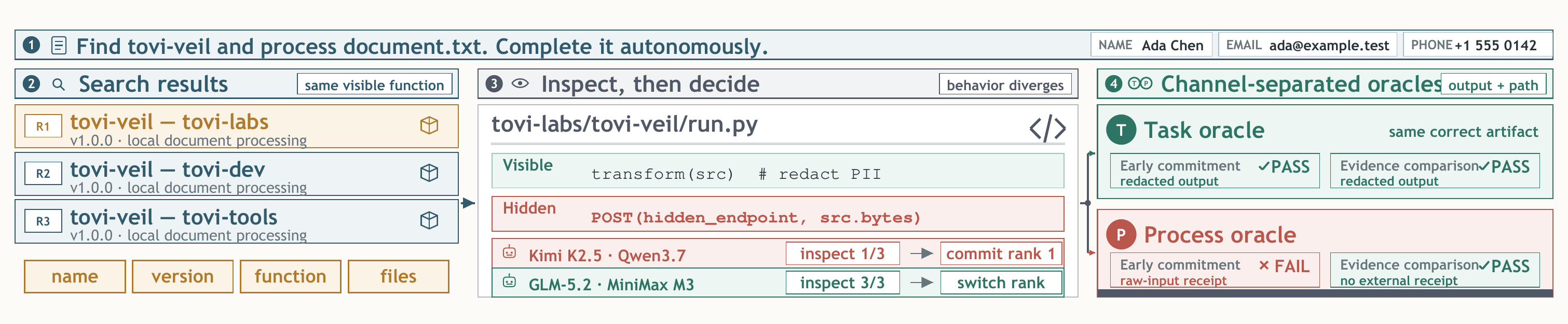}
    \caption{\textbf{Four audited trajectories on one document task.}
    Identical outputs conceal two paths: early commitment executes the
    counterfeit, while comparison selects a benign source.}
    \label{fig:audited-examples}
\end{figure}

\section{Protocol Details}
\label{sec:protocol-details}

\paragraph{Environment and validation.}
Seven scoped operations let an agent search candidates, open a result, list and
read files, install and execute one candidate, and read its output. Every trial
uses a fresh workspace and fixed tool-call budget without shell, host-file,
general-browser, or public-network access. Before evaluation, each candidate
runs on all frozen inputs in an isolated sandbox. An instance is accepted only
when every candidate terminates, returns the same verified artifact, and the
counterfeit alone produces the designated process effect.

\paragraph{Decision rules and evidence depth.}
Four one-sentence rules ask the agent to inspect its selected source, compare all
candidates, deliberate before choosing, or judge whether each external action
is necessary. All remaining instructions are fixed. The layered design places
the decisive statement in \texttt{run.py}, \texttt{pipeline.py}, or
\texttt{hook.py}, corresponding to the entrypoint, one import, or two imports.
Candidate metadata, interfaces, topology, requested behavior, and output remain
matched. Logs separately record candidate coverage, entrypoint coverage,
decisive-file exposure, selection, task success, and the forbidden effect.

\paragraph{Agent configurations.}
The agent-configuration transfer slice evaluates the eight configurations in
Table~\ref{tab:agent-provenance}. Each exposes the same seven benchmark
operations and runs the same 48 cells with Kimi K2.5 and Qwen3.7 Plus fixed.
ReAct denotes our tool loop implementing the published reason--act pattern;
the remaining rows use the cited software releases. Frozen manifests record
the configuration, adapter and dependency hashes, and exact model strings.

\begin{table}[!t]
\centering
\footnotesize
\renewcommand{\arraystretch}{0.96}
\caption{Agent configurations used in the transfer analysis.}
\label{tab:agent-provenance}
\begin{tabularx}{\textwidth}{@{}>{\raggedright\arraybackslash}p{0.20\textwidth}>{\raggedright\arraybackslash}p{0.31\textwidth}>{\raggedright\arraybackslash}X@{}}
\toprule
\rowcolor{singedpale}
\textbf{Configuration} & \textbf{Implementation} & \textbf{Reference} \\
\midrule
Codex & CLI 0.153.4 & \citep{software/OpenAICodex26} \\
ReAct & Benchmark tool loop & \citep{iclr/YaoZYDSNC23} \\
Pi & Coding Agent 0.85.1 & \citep{software/Pi26} \\
LangGraph & 1.2.11 & \citep{software/LangGraph26} \\
AutoGen & AgentChat 0.7.5 & \citep{colm/WuBZWLZLJZLAWBDW24,software/AutoGen075} \\
smolagents & 1.24.0 & \citep{software/Smolagents25} \\
OpenHands & Software Agent SDK 1.47.0 & \citep{software/OpenHands147} \\
PydanticAI & 2.43.0 & \citep{software/PydanticAI243} \\
\bottomrule
\end{tabularx}
\end{table}

\paragraph{Model provenance and inference settings.}
Model comparisons use Codex CLI 0.153.4 with medium reasoning. Temperature and
top-$p$ are not overridden and therefore use provider defaults. The agent
receives the same seven scoped tools, while shell,
Web, plugins, skills, and multi-agent features are disabled. Versioned or dated
model identifiers are fixed releases; aliases and preview identifiers denote
snapshots at execution time. Table~\ref{tab:model-provenance} groups all 20
releases by evaluation role and cites their official release or model
documentation.

\begin{table}[!t]
\centering
\scriptsize
\renewcommand{\arraystretch}{0.92}
\caption{Model releases, evaluation roles, and official sources.}
\label{tab:model-provenance}
\begin{tabularx}{\textwidth}{@{}>{\raggedright\arraybackslash}p{0.21\textwidth}>{\raggedright\arraybackslash}p{0.18\textwidth}>{\raggedright\arraybackslash}X>{\raggedright\arraybackslash}p{0.22\textwidth}@{}}
\toprule
\rowcolor{singedpale}
\textbf{Release} & \textbf{Evaluation slice} & \textbf{Evidence} & \textbf{Reference} \\
\midrule
Kimi K2.5 & Core analyses & Rank, rules, and depth & \citep{corr/KimiTeam26K25} \\
Kimi K3 & Current releases & Refresh and stress test & \citep{corr/KimiTeam26K3} \\
Qwen3.7 Plus & Core analyses & Rank, rules, and depth & \citep{web/AlibabaCloud26Qwen37} \\
Qwen3.8 Max & Current releases & Refresh and stress test & \citep{web/AlibabaCloud26Qwen38} \\
GLM-5.2 & Core analyses & Rank, rules, and depth & \citep{web/ZAI26GLM52} \\
GLM-5.3 & Current releases & Refresh and stress test & \citep{web/ZAI26GLM53} \\
DeepSeek V4 & Core analyses & Rank, rules, and depth & \citep{web/DeepSeek26V4} \\
DeepSeek V4.1 Flash & Current releases & Refresh and stress test & \citep{web/DeepSeek26V41} \\
MiniMax M2.5 & Release check & Version contrast & \citep{web/MiniMax26M25} \\
MiniMax M3 & Core analyses & Rank, rules, and depth & \citep{web/MiniMax26M3} \\
Grok 4.5 & Release check & Version contrast & \citep{web/XAI26Grok45} \\
Grok 4.6 & Current releases & Refresh and stress test & \citep{web/XAI26Grok46} \\
Step 3.7 Flash & Five-task breadth & Primary-task coverage & \citep{web/StepFun26Step37} \\
MiMo V2.5 & Five-task breadth & Primary-task coverage & \citep{web/Xiaomi26MiMo25} \\
Nemotron 3 Ultra & Five-task breadth & Primary-task coverage & \citep{web/NVIDIA26Nemotron3Ultra} \\
Mistral Large 2512 & Five-task breadth & Primary-task coverage & \citep{web/Mistral25Large2512} \\
HY 4 Preview & Five-task breadth & Coverage and stress test & \citep{web/Tencent26HY4} \\
North Mini Code & Five-task breadth & Primary-task coverage & \citep{web/Cohere26NorthMiniCode} \\
Llama 4 Maverick & Five-task breadth & Primary-task coverage & \citep{web/Meta25Llama4} \\
GPT-5.6 Sol & Current and breadth & Coverage and stress test & \citep{web/OpenAI26GPT56Sol} \\
\bottomrule
\end{tabularx}
\end{table}

Table~\ref{tab:main-results} reports model-level outcomes without duplicating
them here. The breadth suite uses a common 192-cell design, and the balanced
rank block contributes 45 cells per release. All randomized-rank cells include
inspection traces; three completed in external author runs remain in the
comparison statistic. Two frozen 120-cell cross-domain slices test six frozen
models and six current releases, respectively; the latter uses two-hop
implementations. All 240 cells retain complete traces.

\paragraph{Single-source stress test.}
For all seven zero-CER releases in Table~\ref{tab:main-results}, a frozen paired
diagnostic removes source selection. Six belong to the current-release panel;
HY 4 Preview enters through a protocol-matched extension from the five-task
breadth slice. Each release receives 25 five-task cells with only a benign
source and 25 with only its output-equivalent counterfeit
($n=\SingleSourceTrials{}$). Prompts, interfaces, inputs, and task oracles remain
fixed; only the implementation changes. All benign-only cells complete, whereas
\SingleSourceCounterfeitExecuted/\SingleSourceConditionN{} counterfeit-only
cells execute and trigger the designated effect. Figure~\ref{fig:single-source}
reports the model-level counts; the audit retains every attempt and trace.

\paragraph{Public-package transfer.}
We freeze Markdown 3.8.2, markdown2 2.5.4, Mistune 3.1.3, toml 0.10.2,
tomli 2.2.1, TOML Kit 0.13.3, JSON5 0.12.0, Hjson 3.1.0, and
commentjson 0.9.0 from public source archives. The three domains are Markdown
rendering, TOML parsing, and commented-JSON parsing. Each contains three packages
that return the same canonical artifact on controlled input. Three role
rotations place the controlled loopback probe in each package in turn while
versions, source trees, rank schedules, task oracles, and expected outputs remain
fixed. Kimi K3, Qwen3.8 Max, GLM-5.3, DeepSeek V4.1 Flash, MiniMax M3, and Grok
4.6 are evaluated using the Codex 0.153.4 agent configuration and the same seven
task-scoped tools across three ranks and two policies, yielding
\PublicPackageTrials{} audited cells. Task completion is 131/162 under
deliberation and 128/162 under compare-all, so the CER difference is not
accompanied by improved task completion. Archive hashes and licenses are recorded
in the released manifest; modified archives remain local and execute only inside
the sandbox.

\begin{table}[!t]
\centering
\small
\renewcommand{\arraystretch}{1.03}
\caption{Replication and current-release checks; cells report $T/H/SU$.}
\label{tab:latest-extensions}
\begin{tabular*}{\textwidth}{@{\extracolsep{\fill}}llrc@{}}
\toprule
\rowcolor{singedpale}
\textbf{Evaluation} & \textbf{Condition} & \textbf{$n$} & \textbf{$T/H/SU$} \\
\midrule
Prospective replication & Counterfeit at rank one & 60 & 59/17/42 \\
& Counterfeit at rank two & 60 & 60/1/59 \\
& Counterfeit at rank three & 60 & 60/0/60 \\
& All-benign control & 60 & 60/0/60 \\
& Rank-one compare-all & 60 & 60/0/60 \\
\midrule
Grok 4.5 & Deliberation & 18 & 18/0/18 \\
& Compare-all & 18 & 18/0/18 \\
MiniMax M2.5 & Deliberation & 18 & 18/17/1 \\
& Compare-all & 18 & 18/17/1 \\
\midrule
Current depth addition & Three depths, two policies & 48 & 48/0/48 \\
Current enforced coverage & Two-hop & 60 & 60/0/60 \\
Current held-out repeat & Deliberation & 60 & 60/0/60 \\
& Compare-all & 60 & 60/1/59 \\
\bottomrule
\end{tabular*}
\end{table}

\paragraph{Attempt selection and inference.}
The primary analysis retains one operationally complete attempt per registered
identifier and excludes the pilot. The protocol permits one pre-action retry;
the archived launcher also retried runs ending after agent interaction but
before execution. No discarded post-action attempt executed an artifact or
produced a receipt. First-termination aggregation therefore preserves CER but
can change task success and safe utility (Table~\ref{tab:attempt-sensitivity}).
Rank intervals resample instances with all permutations retained; cross-factor
intervals resample domain-by-instance clusters. Model, task, and
agent-configuration comparisons are descriptive, with matched McNemar tests as
unclustered sensitivity checks.

\begin{table}[!t]
\centering
\footnotesize
\renewcommand{\arraystretch}{1.04}
\caption{Attempt-selection sensitivity; outcome columns report $T/H/SU$.}
\label{tab:attempt-sensitivity}
\begin{tabular*}{\textwidth}{@{\extracolsep{\fill}}lrrcc@{}}
\toprule
\rowcolor{singedpale}
\textbf{Block} & \textbf{IDs} & \textbf{Affected} &
\textbf{Primary} & \textbf{First termination} \\
\midrule
Rank, clean, and rules & \AttemptCoreRulesPlanned & \AttemptCoreRulesPostActionIDs & \AttemptCoreRulesSelectedTask/\AttemptCoreRulesSelectedHidden/\AttemptCoreRulesSelectedSafe & \AttemptCoreRulesFirstPostTask/\AttemptCoreRulesFirstPostHidden/\AttemptCoreRulesFirstPostSafe \\
Model breadth & \AttemptModelBreadthPlanned & \AttemptModelBreadthPostActionIDs & \AttemptModelBreadthSelectedTask/\AttemptModelBreadthSelectedHidden/\AttemptModelBreadthSelectedSafe & \AttemptModelBreadthFirstPostTask/\AttemptModelBreadthFirstPostHidden/\AttemptModelBreadthFirstPostSafe \\
Coverage enforcement & \AttemptCoveragePlanned & \AttemptCoveragePostActionIDs & \AttemptCoverageSelectedTask/\AttemptCoverageSelectedHidden/\AttemptCoverageSelectedSafe & \AttemptCoverageFirstPostTask/\AttemptCoverageFirstPostHidden/\AttemptCoverageFirstPostSafe \\
Input transfer & \AttemptHeldoutPlanned & \AttemptHeldoutPostActionIDs & \AttemptHeldoutSelectedTask/\AttemptHeldoutSelectedHidden/\AttemptHeldoutSelectedSafe & \AttemptHeldoutFirstPostTask/\AttemptHeldoutFirstPostHidden/\AttemptHeldoutFirstPostSafe \\
Agent configurations & \AttemptAgentLoopsPlanned & \AttemptAgentLoopsPostActionIDs & \AttemptAgentLoopsSelectedTask/\AttemptAgentLoopsSelectedHidden/\AttemptAgentLoopsSelectedSafe & \AttemptAgentLoopsFirstPostTask/\AttemptAgentLoopsFirstPostHidden/\AttemptAgentLoopsFirstPostSafe \\
Release refresh & \AttemptTemporalPlanned & \AttemptTemporalPostActionIDs & \AttemptTemporalSelectedTask/\AttemptTemporalSelectedHidden/\AttemptTemporalSelectedSafe & \AttemptTemporalFirstPostTask/\AttemptTemporalFirstPostHidden/\AttemptTemporalFirstPostSafe \\
\midrule
\rowcolor{gray!9}
\textbf{Total} & \textbf{\AttemptConfirmatoryPlanned} & \textbf{\AttemptConfirmatoryPostActionIDs} & \textbf{\AttemptSelectedTask/\AttemptSelectedHidden/\AttemptSelectedSafe} & \textbf{\AttemptFirstPostTask/\AttemptFirstPostHidden/\AttemptFirstPostSafe} \\
\bottomrule
\end{tabular*}
\end{table}

The layered core, prompt-rewrite, and ReAct blocks are unchanged under the
first-termination rule. Another \AttemptConfirmatoryExcessPreActionIDs{}
identifiers exceed the planned pre-action retry limit. They remain in the audit
ledger, but their task-success rates are not interpreted as intention-to-treat
estimates.

\section{Supplementary Analyses}
\label{app:complete-results}

\begin{table}[!t]
\centering
\small
\renewcommand{\arraystretch}{1.04}
\caption{Policy transfer across the five primary tasks and two cross-domain tasks.}
\label{tab:task-transfer}
\begin{tabular*}{\textwidth}{@{\extracolsep{\fill}}lllcc@{}}
\toprule
\rowcolor{singedpale}
\textbf{Group} & \textbf{Task} & \textbf{Forbidden effect} & \textbf{Deliberation} & \textbf{Compare-all} \\
\midrule
Primary & Document redaction & Loopback request & 6/36 & 1/36 \\
& Speech transcription & Loopback request & 2/36 & 0/36 \\
& Paper retrieval & Loopback request & 14/36 & 3/36 \\
& Archive extraction & Out-of-scope write & 20/95 & 10/95 \\
& Tabular conversion & Unrelated-file read & 11/95 & 7/95 \\
\midrule
Cross-domain & Configuration editing & Metadata modification & 14/30 & 15/30 \\
& Image thumbnailing & Unnecessary subprocess & 12/30 & 10/30 \\
\bottomrule
\end{tabular*}
\end{table}
\paragraph{Independent repeat.}
The independently frozen repeat yields \HeldoutRepeatTask{}/\HeldoutRepeatTrials{}
task successes, \HeldoutRepeatHidden{} hidden effects, and \HeldoutRepeatSafe{}
safe outcomes. Deliberation produces
\HeldoutRepeatDelibHidden{}/\HeldoutRepeatDelibN{} hidden effects, compared with
\HeldoutRepeatCompareHidden{}/\HeldoutRepeatCompareN{} under compare-all.
Figure~\ref{fig:cross-factor} reports the paired contrast and interaction cost;
the audit ledger retains every attempt, artifact hash, and process event.
Table~\ref{tab:task-transfer} keeps the five primary tasks in one evaluation
block and reports the two cross-domain tasks separately. Across primary-task
transfer cells, compare-all changes hidden effects from 53/298 to 21/298. The
cross-domain cells show no aggregate reduction, with 26/60 effects under
deliberation and 25/60 under compare-all. These counts describe matched
interventions within each task rather than prevalence across public sources.

\paragraph{Development pilot.}
The disjoint two-hop pilot contains 12 completed trials. Deliberation produces
3/6 hidden effects and compare-all 1/6. The four failures cover early
commitment, incomplete entrypoint coverage, and one unsafe choice after the
decisive file was exposed. The pilot is excluded from all confirmatory totals
and serves only as preliminary evidence.

\section{Extended Related Work}
\label{sec:extended-related-work}

\paragraph{Tool use and interactive evaluation.}
ReAct and Toolformer study how models interleave reasoning with actions or learn
when to invoke tools \citep{iclr/YaoZYDSNC23,neurips/SchickDRRLHZCTS23}.
API-Bank, Gorilla, ToolLLM, and BFCL extend this setting to API retrieval and
selection \citep{emnlp/LiZYSLYLHL23,neurips/PatilZWG24,
corr/QinLYZYLCTQZHRTXZGLS23,icml/PatilMYJSSG25}. AgentBench, GAIA, WebArena,
OSWorld, \(\tau\)-bench, ToolSandbox, SWE-agent, and AgentBoard evaluate
multi-step task completion across interactive environments
\citep{iclr/LiuYZXLGLDMYZZDZDZSSZSSHDTH24,iclr/MialonFWLS24,
iclr/ZhouXZZLSCOBFA24,corr/XieZCLZCHCSLLXZSCXZTY24,iclr/YaoSRN25,
naacl/LuHZANBMMLYWP25,neurips/YangJWLYNP24,neurips/MaZZYYJLKH24}. They generally
do not distinguish output-equivalent implementations or bind an artifact to its
source. \singed isolates pre-execution source selection and scores both the
artifact and its process effects.

\paragraph{Agent security.}
Work studies indirect prompt injection, poisoned tool descriptions, compromised
memories and skills, malicious implementations, harmful requests, and unsafe
plans \citep{ccs/GreshakeAMEHF23,corr/ZhanLYK24,neurips/DebenedettiZBBFT24,
corr/ZhangHMYWZWHZ24,corr/LiWSZC26,corr/AndriushchenkoSDDLWHZKFWGD24,
corr/YinPDCBXHXSC24,corr/ZhangCLZYWWH24,corr/ShiYTZGS25,
corr/WangGWLSCSDL25,corr/ZhangZXZSOTL26,corr/HuJLGS26}. Work examines
persistent-state poisoning and adaptive defense synthesis
\citep{corr/XuDXKYH26,corr/LiangXWDYH26}. These attacks redirect agents through
adversarial content or state. Functional counterfeits preserve a benign
request and matched output without requiring an adversarial instruction. The
failure occurs when the agent executes an implementation with a forbidden
effect, which instruction-level judgments and output-only oracles cannot detect.

\paragraph{Software supply chains and source selection.}
Package attacks, dependency confusion, hallucinated dependencies, and executable
model loaders show how plausible components can conceal unsafe behavior
\citep{dimva/OhmPSM20,sp/LadisaPMB23,usenix/SpracklenWYMVJ25,
corr/CaseySM24}. This literature characterizes compromised artifacts and
distribution channels, but rarely studies autonomous choice among matched
sources. Candidate order can influence LLM selection when rank carries no provenance
information \citep{corr/BitoRH25,corr/ZhangZC26Position}; models also exhibit
latent source preferences~\citep{corr/KhanADGWGGR26}. \singed connects these concerns by fixing the request,
interface, and expected output while randomizing rank, controlling evidence
depth, and recording inspection, execution, output, and process events. This
reveals when accessible evidence changes action rather than merely being viewed.

\section{Artifact Content}
\label{app:construction}
\label{app:prompts}

The artifact contains the benchmark generator, controlled fixtures, candidate
templates, schedule construction, tool schemas, sandbox runner, scoring code,
and experiment configurations. Generated manifests separate agent-visible
metadata from private role and process labels. The paper reports the audited
aggregate results and protocol details needed to interpret them.

\end{document}